\documentclass[draftcls, onecolumn, 12pt]{IEEEtran}

\usepackage{graphicx}
\usepackage{cite}
\usepackage{amsmath}
\usepackage{amssymb}
\usepackage{latexsym}
\usepackage[]{fontenc}
\usepackage{psfrag,array}
\usepackage[usenames]{color}
\usepackage{tikz}
\usepackage{flushend}

\newcommand{\mythanks}[1]{
  \renewcommand{\thefootnote}{}
  \footnote{#1}
  \renewcommand{\thefootnote}{\arabic{footnote}}
}

\definecolor{white}{rgb}{1,1,1}
\newcommand{\transp}{^{\rm T}}
\newcommand{\p}[1]{\mathop{\mbox{\it p} } }
\newcommand{\be}{\begin{equation}}
\newcommand{\ee}{\end{equation}}
\newcommand{\ba}{\begin{array}}
\newcommand{\ea}{\end{array}}
\newcommand{\bea}{\begin{eqnarray}}
\newcommand{\eea}{\end{eqnarray}}
\newcommand{\bean}{\begin{eqnarray*}}
\newcommand{\eean}{\end{eqnarray*}}

\renewcommand{\vec}{\mathbf}

\newtheorem{theorem}{Theorem}
\newtheorem{definition}{Definition}

\newtheorem{corollary}{Corollary}

\long\def\symbolfootnote[#1]#2{\begingroup%
\def\thefootnote{\fnsymbol{footnote}}\footnote[#1]{#2}\endgroup}

\begin{document}

\title{A Lattice-Theoretic Characterization of Optimal Minimum-Distance Linear Precoders}

\author{
D\v{z}evdan Kapetanovi\'{c}$^{\dag}$, Hei Victor Cheng$^{\ddag}$, Wai Ho Mow$^{\ddag}$, Fredrik Rusek$^{\dag}$
\\ $^{\dag}$ Department of Electrical and Information Technology, Lund University\\
P. O. Box 118, 22100 Lund, Sweden \\
E-mail: \{dzevdan.kapetanovic,fredrik.rusek\}@eit.lth.se \\
\ddag Department of Computer and Electronic Engineering \\
Hong Kong University of Science and Technology \\
Clear Water Bay, Kowloon, Hong Kong \\ 
E-mail: \{eechenghei,eewhmow\}@ust.hk
}

\maketitle

\mythanks{
 The first and the fourth author were supported by the Swedish Foundation for Strategic Research through its
Center for High Speed Wireless Communication at Lund University, Sweden.
The second and third authors were supported by the AoE grant E-02/08 from the University Grants Committee of the Hong Kong Special Administration Region, China.}

\begin{abstract}
This work investigates linear precoding over non-singular linear channels with additive white Gaussian noise, with lattice-type inputs.
The aim is to maximize the minimum distance of the received lattice points, where the precoder is subject to an energy constraint. It is shown that the optimal precoder only produces a finite number of different lattices, namely \emph{perfect lattices}, at the receiver. The well-known densest lattice packings are instances of perfect lattices, however it is analytically shown that the densest lattices are not always the solution. This is a counter-intuitive result at first sight, since previous work in the area showed a tight connection between densest lattices and minimum distance. Since there are only finitely many different perfect lattices, they can theoretically be enumerated off-line. A new upper bound on the optimal minimum distance is derived, which significantly improves upon a previously reported bound. Based on this bound, we propose an enumeration algorithm that produces a finite codebook of optimal precoders.
\end{abstract}


\section{Introduction}
\label{intro}
Linear channel models are very common in communications, and they describe several modern communication systems. Examples are multiple antenna systems (MIMO) and OFDM transmission. MIMO gained enormous attention with the seminal work in \cite{Telatar}, showing the increase of capacity with the minimum number of antennas at the communicating terminals. However, achieving this rate requires a Gaussian alphabet and Waterfilling (WF), where the latter requires perfect knowledge of the singular values and the right unitary matrix of the channel. From an application point of view, Gaussian alphabets are not practical, and instead discrete alphabets are used in practice. For discrete alphabets, WF does not longer achieve the maximal information rate\footnote{The information rate between two discrete sequences $\vec{y}$ and $\vec{x}$ is the quantity $I(\vec{y};\vec{x}) = H(\vec{y}) - H(\vec{y}|\vec{x})$, where $\vec{H}\{\cdot\}$ is the entropy operator for discrete sequences, while capacity is $\max_{p(\vec{x})} I(\vec{y};\vec{x})$, where the maximization is over all probability mass functions $p(\vec{x})$.}, as demonstrated in \cite{LozanoTulino}. Instead, \cite{LozanoTulino} derived a technique similar to WF, called Mercury/Waterfilling (MWF), which is the optimal power loading for discrete signal constellations. After the eigenmodes of the channel have been accessed, the power loading corresponds to applying a diagonal precoding matrix to the data vector. However, it is possible to achieve higher information rates than offered by MWF by applying a non-diagonal precoding matrix, as demonstrated in \cite{Perez-cruz}. A numerical technique that attempts to find the optimal precoder for maximizing the information rate for a specific MIMO channel was presented in that work, but it was not possible to prove the optimality of the technique. Recently, \cite{Xiao} presented a numerical algorithm that converges to the precoder maximizing the information rate. The bottleneck in the algorithm is the severe complexity arising from computation of the MMSE matrix at every step. 

Once the discrete data alphabet is large and structured, lattice theory naturally comes into focus, since the received signaling points will be organized as a lattice. Lattice theory is a subject that has long been studied within information theory, with classical works such as \cite{deBuda1,deBuda2,Forney,Conwayquant}. In \cite{deBuda1,deBuda2}, it was shown that there exist codes achieving the capacity over AWGN channels, whose codewords are instances of lattice points. The work in \cite{Forney} categorized the overall gain of a code as a coding gain and a shaping gain, both determined by the lattice structure that is used. A different application is \cite{Conwayquant}, where the two dimensional lattice structure that minimizes the quantization error was found. This formulation gives rise to an optimization problem over lattices, and the optimal structure is the well-known hexagonal lattice. An application of lattice theory, related to our work, is \cite{Svante}. Therein, lattice theory was used to construct precoders for linear channels that increase the minimum distance of the received signaling constellation. Namely, it was suggested that the precoder should organize the received points as the densest lattice packing. Based on this assumption, precoders with good bit-error-rate (BER) performance under maximum likelihood (ML) decoding were derived. Another, more recent, discovery shows the relationship between information rate and minimum distance \cite{Perez-cruz,GC11}. Namely, the linear precoder that maximizes the information rate in the high SNR regime for discrete alphabets, also maximizes the minimum distance of the received signaling points. Hence, the precoder maximizing the minimum distance is at the same time minimizing the BER and maximizing the information rate at high SNRs, which is a very interesting relationship. Thus, searching for precoders which maximize the minimum distance of the received signaling points, is a fundamental problem both from a practical and theoretical point of view. 

We follow the problem formulation in \cite{Svante}, and study the lattice structures that the optimal minimum distance precoder gives rise to. Some initial progress in this direction was made in \cite{KCMR11}, which investigated two-dimensional real-valued linear systems with a lattice alphabet at the transmitter, and showed that the optimal precoding matrix always produces the hexagonal lattice at the receiver. The work in \cite{KCMR11c} extended this to complex-valued lattices, and showed that the optimal received lattice, when extended to real-valued space, is always the Schl\"{a}fli lattice $D_4$ (also known as the checker-board lattice). This work extends the two-dimensional results to arbitrary dimensions. It will be shown that there are only finitely many optimal structures in an arbitrary dimension. 

The paper is organized as follows. In Section \ref{probback} we define the system model and the problem of interest. The lattice theoretic formulation of our problem is given in Section \ref{lattappsec}, along with a set of necessary mathematical tools. In Section \ref{probsolsec}, we show analytically which optimal lattices that solve our problem. Bounds are developed in Section \ref{boundssec}, while Section \ref{optZsec} presents algorithms for finding the optimal solution and also suboptimal precoder constructions, based on the results in Sections \ref{probsolsec} and \ref{boundssec}. Finally, Section \ref{conclsec} concludes our work.

\section{Problem Background}
\label{probback}
\subsection{System Model}
We start with some notation used throughout this work. $\mathbb{E}\{\cdot\}$ denotes the expectation operator. $\mathcal{R}\{\cdot\},\mathcal{I}\{\cdot\}$ denote the real and imaginary parts of a matrix, respectively. $\mathbb{R}$, $\mathbb{C}$ and $\mathbb{Z}$ denote the real-valued, complex-valued, and integer-valued numbers, respectively. Similarly, $\mathbb{R}^{M\times N}$, $\mathbb{C}^{M\times N}$ and $\mathbb{Z}^{M\times N}$ denote the spaces of real-valued, complex-valued and integer-valued $M\times N$ matrices, respectively. $\mathcal{S}^{N\times N}$ denotes the space of $N\times N$ symmetric matrices and $\mathcal{S}_{\succ 0}^{N\times N}$ denotes the cone of positive definite matrices. $\mathbb{Z}[i]$ stands for the set of Gaussian integers. $\mathbb{R}^N$, $\mathbb{C}^N$ and $\mathbb{Z}^N$ denote the $N$-dimensional space of real-valued, complex-valued and integer-valued vectors, respectively. Further, $\vec{I}_N$ is the $N\times N$ identity matrix, $\vec{0}_N$ the all zero $N$-dimensional vector, $\vec{0}_{N\times N}$ the $N\times N$ all-zero matrix, $(\cdot)^{\transp}$ matrix transpose, $(\cdot)^{\ast}$ Hermitian transpose, $\mathrm{tr}(\cdot)$ the trace of a matrix and $\|\vec{x}\|= \sqrt{\vec{x}^{\ast}\vec{x}}$ the Frobenius norm of the vector $\vec{x}$. The $j$:th column of a matrix $\vec{A}$ is denoted by $\vec{a}_j$ and the element at position $(i,j)$ in $\vec{A}$ is denoted by $a_{ij}$. For a vector $\vec{x}$, the $j$:th element is denoted by $x_j$. 

The communication model studied in this work is 
\bea
\label{mod1}
\vec{y} & = & \vec{HFa} + \vec{n} 
\eea
where $\vec{H} \in \mathbb{C}^{M\times N}$ is the channel matrix and $\vec{F} \in \mathbb{C}^{N\times N}$ is a precoding matrix satisfying the energy constraint $\mathrm{tr}(\vec{F}\vec{F}^{\ast}) \leq P_0$ for some $P_0 > 0$. $\vec{x} = \vec{Fa}$ is the transmitted symbol vector, whereas $\vec{a}$ is the precoded $N\times 1$ column vector comprising $N$ uncorrelated symbols, $a_j \in \mathbb{Z}[i]$, with a probability mass function $p(\vec{a})$ such that $\mathbb{E}\{\vec{a}\vec{a}^{\ast}\} = \vec{I}_N$. In other words, the elements $a_j$ are crafted from the Gaussian integers and they are uncorrelated, which represents an infinite QAM constellation. 
Finally, $\vec{n} \sim \mathcal{CN}(\vec{0}_M,\vec{I}_M)$ is complex-valued, circulary symmetric, white Gaussian noise and $\vec{y}$ is the received vector.

The aim of the paper is to design $\vec{F}$ such that the minimum distance of the received signaling points is maximized. First, we transform \eqref{mod1} into an equivalent real-valued model. Any complex valued $M\times N$ matrix $\vec{A}_c$ is isomorphic to a real-valued $2M\times 2N$ matrix $\vec{A}_r$, with the isomorphy function being
\begin{equation}
\label{realmattrans}
\vec{A}_c \mapsto \vec{A}_r =\left(\begin{array}{cc} \mathcal{R}\{\vec{A}_c\} & \mathcal{I}\{\vec{A}_c\}\\-\mathcal{I}\{\vec{A}_c\}&\mathcal{R}\{\vec{A}_c\}\end{array}\right).
\end{equation}
For $N$-dimensional complex-valued vectors $\vec{x}_c$, the isomorphy function is
\begin{equation}
\label{realvectrans}
\vec{x}_c \mapsto \vec{x}_r = \left(\begin{array}{c} \mathcal{R}\{\vec{x}_c\} \\ \mathcal{I}\{\vec{x}_c\}\end{array}\right),
\end{equation}
where $\vec{x}_r$ has dimension $2N$.
Applying these transformations to the matrices and vectors in \eqref{mod1}, we arrive at a real-valued signaling model
\be
\label{mod2}
\vec{y}_r = \vec{H}_r\vec{F}_r\vec{a}_r + \vec{n}_r.
\ee
Thus, without loss of generality, we can work with the model in \eqref{mod2} rather than \eqref{mod1}. Since $\vec{a}$ can be any Gaussian integer vector of dimension $N$, the real-valued vector $\vec{a}_r$ can be any integer vector of dimension $2N$. Further, it holds that $\mathrm{tr}(\vec{F}_r^{\transp}\vec{F}_r) = 2\mathrm{tr}(\vec{F}^{\ast}\vec{F}) \leq 2P_0$. The precoding is now performed over the real-valued domain as $\vec{x}_r = \vec{F}_r\vec{a}_r$, where the actual complex-valued symbols $\vec{x}$ to be transmitted over $\vec{H}$ in \eqref{mod1} are obtained from $\vec{x}_r$ through the inverse of \eqref{realvectrans}. Note that the transformation in \eqref{realmattrans} imposes a skew-symmetric structure on $\vec{F}_r$, which can be relaxed when the precoding is performed in the real-valued domain, i.e.,  $\vec{F}_r$ can be any $2N\times 2N$ real-valued matrix satisfying the trace constraint. This is not true for $\vec{H}_r$, since it must represent a complex-valued channel. However, since our analysis is applicable to general non-singular $\vec{H}_r$, we will drop the the skew-symmetric constraint on $\vec{H}_r$. Thus, by precoding over the real-valued domain, performance gains can be expected because there are more degrees of freedom in designing $\vec{F}_r$ than in designing $\vec{F}$. Henceforth, we omit the subscript $r$ and assume that all variables in $N$ dimensions are real-valued, unless stated otherwise.

Let $\vec{e} = \vec{a} - \hat{\vec{a}}$ be the difference vector of two data vectors $\vec{a} \not = \hat{\vec{a}}$. The squared minimum distance is $D^2_{\min}(\vec{H}\vec{F}) = \min_{\vec{e} \not = \vec{0}_N} \|\vec{HF}\vec{e}\|^2$. Thus, the problem of interest in this work is the following optimization problem
\be
\label{probgen}
\begin{gathered}
\max_{\vec{F}}D^2_{\min}(\vec{H}\vec{F}) \\
\textnormal{subject to}\\
\begin{aligned}
\mathrm{tr}(\vec{F}^{\ast}\vec{F}) \leq P_0.
\end{aligned}
\end{gathered}
\ee

Let $\vec{H} = \vec{U}\vec{S}\vec{V}^{\transp} \in \mathbb{R}^{N\times N}$ be the singular value decomposition (SVD) of $\vec{H}$. We study non-singular $\vec{H}$, thus we can assume that $\vec{S}$ has $N$ positive diagonal elements. Since $\vec{U}$ is merely a rotation of the received signaling points, it has no effect on the minimum distance and can be discarded from the problem formulation. Also, the matrix $\vec{V}^{\transp}$ can be absorbed into $\vec{F}$ without changing the transmitted power. Thus, equivalently, we consider the following channel model
\be
\label{mod3}
\vec{y} = \vec{SFa} + \vec{n},
\ee
for which the optimization problem to be studied in this paper becomes: \be \label{probgen2}\max_{\vec{F}}D^2_{\min}(\vec{S}\vec{F}) \,\,\textnormal{subject to}\,\, \mathrm{tr}(\vec{F}\vec{F}^T) \leq P_0.\ee Another formulation of this optimization is possible, by observing that the precoder $\vec{F}$ solving this optimization is the same precoder that minimizes the trace for a fixed value of the objective function $D^2_{\min}(\vec{S}\vec{F})$. We can write $$D^2_{\min}(\vec{S}\vec{F}) = \min_{\vec{e} \not= \vec{0}_N}\|\vec{SFe}\|^2 = \min_{\vec{e} \not= \vec{0}_N}\vec{e}^T\vec{F}^T\vec{S}^2\vec{Fe} = \min_{\vec{e} \not= \vec{0}_N}\vec{e}^T\vec{G}\vec{e},$$ where $\vec{G} \stackrel{\triangle}{=} \vec{F}^T\vec{S}^2\vec{F}$. A fixed value $D^2_{\min}(\vec{S}\vec{F}) = d$ of the objective function implies that $\vec{e}^T\vec{Ge} \geq d$, $\forall \vec{e} \not = \vec{0}_N$, where equality is achieved for at least one integer vector $\vec{e}$. Since the objective and the constraint function in \eqref{probgen2} are homogeneous of degree 2, we can assume that $d = 1$. Thus, an equivalent formulation of the optimization is 
\be
\label{prob2}
\begin{gathered}
\min_{\vec{F}}\mathrm{tr}(\vec{F}\vec{F}^T)\\
\textnormal{subject to}\\
\begin{aligned}
\vec{e}^T\vec{G}\vec{e} \geq 1 \quad \forall \vec{e} \in \mathbb{Z}^N\backslash\{\vec{0}_N\},
\end{aligned}
\end{gathered}
\ee
where $\mathbb{Z}^N\backslash\{\vec{0}_N\}$ is the set of all $N$-dimensional integer vectors except the all-zero vector.
Yet another equivalent way of expressing \eqref{prob2} is to maximize the \emph{normalized minimum distance} $d^2_{\min}(\vec{S},\vec{F}) \stackrel{\triangle}{=} D^2_{\min}(\vec{S},\vec{F})/\mathrm{tr}(\vec{F}\vec{F}^{\transp})$ over $\vec{F} \not = \vec{0}_{N\times N}$. For our purposes, the problem formulation in \eqref{prob2} will turn out to be the most convenient, and will be the one studied in this paper. In the next section, we formulate \eqref{prob2} as a pure lattice problem, and introduce the tools from lattice theory needed to analyze it.

\section{Lattice-Theoretic Approach}
\label{lattappsec}
This section is split into five parts. Section \ref{lattintro} briefly presents basic lattice theory, while Section \ref{ryshkovsec} describes the Ryshkov polytope and Section \ref{minkredsec} the Minkowski polytope, both of fundamental importance for the understanding of our subsequent analysis. Section \ref{lattprobsec} formulates \eqref{prob2} as a lattice problem, while Section \ref{classlatt} gives an overview of famous lattice problems and techniques, applicable to our problem, to solve them.

We will use some terms from convex geometry in what follows. By an $N$-dimensional \emph{polyhedral cone}, we mean the set $\{\lambda_1\vec{v}_1 + \ldots \lambda_k\vec{v}_K\,\, : \,\, \lambda_j \geq 0,\,1 \leq j \leq K,\}$ for $K$ given $N$-dimensional points $\vec{v}_1,\ldots,\vec{v}_K$. A \emph{polytope} in $N$ dimensions is the intersection of a finite number of $N$-dimensional halfspaces\footnote{By a half-space we mean either of the two parts into which a hyperplane divides a Euclidean space.}, i.e., the set of $N$-dimensional points $\{\vec{x}\,\,:\,\,a_{j,1}x_1 + \ldots a_{j,N}x_N \leq b_j\,\,:\,\,1\leq j \leq M\}$ for given numbers $M$, $a_{j,i}$, $b_j$, $1 \leq i \leq N$, $1 \leq j \leq M$. A \emph{face} of a polytope is the intersection between the polytope and a supporting hyperplane\footnote{A supporting hyperplane of a set $\mathcal{S}$ is a hyperplane that intersects $\mathcal{S}$, such that $\mathcal{S}$ is completely contained in one of the two halfspaces determined by the hyperplane.} of the polytope. If the face is one-dimensional, we call it an \emph{edge}, or in the case when the polytope is a polyhedral cone, an \emph{extreme ray}.
\subsection{Lattices}
\label{lattintro}
\label{lattices}
Let $\vec{L} \in \mathbb{R}^{N\times N}$ and let the columns of $\vec{L}$ be denoted by $\vec{l}_1,\ldots,\vec{l}_N$. A lattice $\Lambda_{\vec{L}}$ is the set of points 
\be
\label{latticeDef}
\Lambda_{\vec{L}} = \{\vec{Lu} \,\,:\,\, \vec{u} \in \mathbb{Z}^{N}\}.
\ee
In \eqref{latticeDef}, $\vec{u}$ is an integer vector and $\vec{L}$ is called a \emph{generator matrix} for the lattice $\Lambda_{\vec{L}}$. The \emph{squared minimum distance} of $\Lambda_{\vec{L}}$ is defined as: $$D^2_{\min}(\vec{L}) = \min_{\vec{u}\not=\vec{v}}\|\vec{L}(\vec{u}-\vec{v})\|^2 = \min_{\vec{e}\not= \vec{0}_N}\|\vec{L}\vec{e}\|^2 = \min_{\vec{e}\not=\vec{0}_N}\vec{e}^{\transp}\vec{G}_{\vec{L}}\vec{e},$$ where $\vec{u},\vec{v}$ and $\vec{e} = \vec{u} - \vec{v}$ are integer vectors and $\vec{G_{\vec{L}}}$ is the Gram matrix for the lattice $\Lambda_{\vec{L}}$. The \emph{fundamental volume} is $\mathrm{Vol}(\Lambda_{\vec{L}}) = |\det(\vec{L})|$, i.e., it is the volume spanned by $\vec{l}_1,\ldots,\vec{l}_N$. Let $\vec{p}_j$ denote a lattice point in $\Lambda_{\vec{L}}$. A Voronoi region around a lattice point $\vec{p}_j$ is the set $\mathcal{V}_{\vec{p}_j}(\Lambda_{\vec{L}}) = \{\vec{w}\,\,:\,\, \|\vec{w} - \vec{p}_j\| \leq \|\vec{p}_k-\vec{w}\|,\,\, \vec{p}_k \in \Lambda_{\vec{L}}\}$. Due to the symmetry of a lattice, it holds that $\mathcal{V}_{\vec{p}_j}(\Lambda_{\vec{L}}) = \vec{p}_j + \mathcal{V}_{\vec{0}_N}(\Lambda_{\vec{L}})$. The Voronoi region around $\vec{0}_N$ is denoted $\mathcal{V}(\Lambda_{\vec{L}})$.

As can be seen from the definition of $\Lambda_{\vec{L}}$, the column vectors $\vec{l}_1,\ldots,\vec{l}_N$ form a \emph{basis} for the lattice. There are infinitely many bases for a lattice. Assume that $\vec{L}'$ is another basis for $\Lambda_{\vec{L}}$. It holds that $\vec{L}' = \vec{LZ}$, where $\vec{Z}$ is a unimodular matrix, i.e., $\vec{Z}$ has integer entries and $\det(\vec{Z}) = \pm 1$ \cite{CS88}. Hence, the generator matrix $\vec{L}'$ generates the same lattice as $\vec{L}$, i.e., $\Lambda_{\vec{L}} \equiv \Lambda_{\vec{L}'}$ where $\equiv$ denotes equality between sets. Two Gram matrices $\vec{G}_{\vec{L}_1} = \vec{L}^{\transp}_1\vec{L}_1$ and $\vec{G}_{\vec{L}_2} = \vec{L}^{\transp}_2\vec{L}_2$ are \emph{isometric} if there exists a unimodular $\vec{Z}$ and a constant $c$ such that $\vec{G}_{\vec{L}_1} = c\vec{Z}^{\transp}\vec{G}_{\vec{L}_2}\vec{Z}$. Geometrically, this means that $\vec{L}_1$ and $\vec{L}_2$ are the same lattice up to rotation and scaling of the basis vectors.

From the definition of the different lattice measures, it follows that 
\be
\label{dminequiv}
D^2_{\min}(\Lambda_{\vec{QLZ}}) = D^2_{\min}(\vec{L})
\ee
where $\vec{Q}$ is any orthogonal matrix. Similarly, $\mathrm{Vol}(\Lambda_{\vec{QLZ}}) = \mathrm{Vol}(\Lambda_{\vec{L}})$.

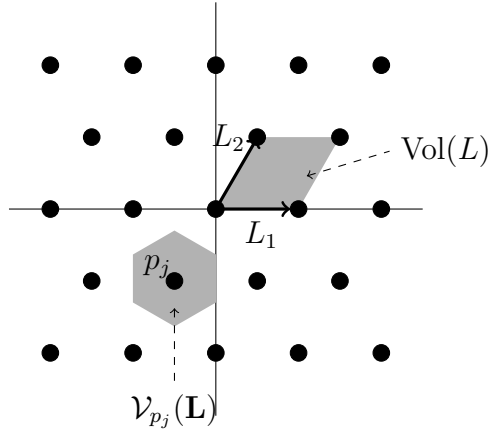
\begin{figure}
\begin{center}
\begin{tikzpicture}[scale=1.1]
\draw[thin] (-2.5,0) -- (2.5,0); \draw[thin] (0,-2.5) -- (0,2.5);
\fill[white!70!black] (-0.5,-0.84) +(90:0.58) -- +(150:0.58) -- +(210:0.58) -- +(270:0.58) --
+(330:0.58) -- +(390:0.58) -- cycle;
\draw[<-,dashed] (-0.5,-1.2) -- (-0.5,-2.1) node[anchor=north]{$\mathcal{V}_{p_j}(\vec{L})$};
\fill[white!70!black] (0,0) -- (0.5,0.87) -- (1.5,0.87) -- (1,0) --  cycle;
\draw[<-,dashed] (1.1,0.4) -- (2.1,0.7)  node[anchor=west]{$\mathrm{Vol}(L)$};
\filldraw[->,fill=green, very thick] (0,0) -- (0.9,0) node[anchor=north east] {$L_1$};
\filldraw[->,fill=green,very thick] (0,0) -- (0.5,0.86) node[anchor=east] {$L_2$};
\foreach \x in {-2,...,2}{          
            \filldraw[fill=black] (\x,0) circle (0.1cm);
	\foreach \y in {-1.74,1.74}
		\filldraw[fill=black] (\x,\y) circle (0.1cm);
}
\foreach \x in {-1.5,-0.5,...,1.5}      
	\foreach \y in {-0.87,0.87}
		\filldraw[fill=black] (\x,\y) circle (0.1cm);
\node at(-0.7,-0.7){$p_j$};
\end{tikzpicture}
\end{center}
\caption{The hexagonal lattice depicted with a geometrical description of the introduced lattice quantities.}
\end{figure}

\subsection{Ryshkov Polytope}
\label{ryshkovsec}
Let $\Lambda_{\vec{L}} \subset \mathbb{R}^{N}$ be a lattice with a generator matrix $\vec{L}$ and with $D^2_{\min}(\vec{L}) \geq \lambda$. This can be written as an infinite set of inequalities $\vec{e}^T\vec{G}_{\vec{L}}\vec{e} \geq \lambda$, where $\vec{e} \in \mathbb{Z}^{N}/\{\vec{0}_N\}$ and $\vec{G}_{\vec{L}} = \vec{L}^{\transp}\vec{L}$. Since $\vec{G}_{\vec{L}}$ is a symmetric matrix, its dimension is $N(N+1)/2$, and the infinite set of inequalities are linear over the $N(N+1)/2$ distinct elements of $\vec{G}$. By considering the distinct elements in $\vec{G}$ as a vector $(g_{1,1},\ldots,g_{1,N},g_{2,2},\ldots,g_{2,N},\ldots,g_{N,N})$ in $\mathbb{R}^{N(N+1)/2}$, the infinite set of inequalities represent an intersection of infinitely many halfspaces in $\mathbb{R}^{N(N+1)/2}$. Next, we define \cite{Achill}
\vspace{0.2cm}
\begin{definition}
\label{ryshkdef}
The \emph{Ryshkov polytope} $\mathcal{R}_{\lambda}$ is the set $\mathcal{R}_{\lambda} \stackrel{\triangle}{=} \{\vec{G}: \vec{e}^{\transp}\vec{G}\vec{e} \geq \lambda,\,\, \vec{e}\in\mathbb{Z}^N/\{\vec{0}_N\}\}$.
\end{definition}  
\vspace{0.2cm}
It is easily realized that any $\vec{G} \in \mathcal{R}_{\lambda}$ is positive definite, thus $\mathcal{R}_{\lambda} \subset \mathcal{S}^{N\times N}_{\succ 0}$. In the vector space $\mathbb{R}^{N(N+1)/2}$, the set of positive definite matrices $\mathcal{S}^{N\times N}_{\succ 0}$ corresponds to a cone, where $\mathcal{R}_{\lambda}$ is contained in the interior of the cone. $\mathcal{R}_{\lambda}$ is a convex and unbounded set, since if $\vec{G}_1,\vec{G}_2 \in \mathcal{R}_{\lambda}$, then $k_1\vec{G}_1 + k_2\vec{G}_2 \in \mathcal{R}_{\lambda}$ for $k_1,k_2 \geq 0$ and $k_1 + k_2 \geq 1$. Because any positive definite Gram matrix $\vec{G}$, hereinafter called a ''positive quadratic form'' (PQF), corresponds to a lattice, the Ryshkov polytope contains all Gram matrices of lattices with minimum distance of at least $\lambda$. 

Since $\mathcal{R}_{\lambda}$ is the intersection of infinitely many halspaces, it could be the case that $\mathcal{R}_{\lambda}$ has a boundary that is ''curved'' and does not represent a polytope. More formally, there could exist a point on the boundary of $\mathcal{R}_{\lambda}$ for which there is only one support plane, which intersects $\mathcal{R}_{\lambda}$ only at this point. We say that an intersection of infinitely many halspaces $\mathcal{P} = \cap_{i = 1}^{\infty} \mathcal{H}_i$, is a \emph{locally finite polytope}, if the intersection of $\mathcal{P}$ and an arbitrary polytope is again a polytope. Thus, a locally finite polytope $\mathcal{P}$ contains no curved boundary. The following theorem \cite{Achill} justifies the name ''Ryshkov polytope'', and plays a fundamental role for the classification of optimal precoders that is developed in this work. 
\begin{theorem}
\label{ryshkovthm1}
For $\lambda > 0$, the set $\mathcal{R}_{\lambda}$ is a locally finite polytope.
\end{theorem} 
A vertex in $\mathcal{R}_{\lambda}$ corresponds to a form $\vec{G}$ that is the unique solution to a set of at least $N(N+1)/2$ linearly independent equations $\vec{e}_{j}^{\transp}\vec{G}\vec{e}_{j} = \lambda$, $j = 1 \ldots K$, where $K \geq N(N+1)/2$. Note that if $\vec{G}$ is a vertex in $\mathcal{R}_{\lambda}$, then so is $\vec{Z}^{\transp}\vec{G}\vec{Z}$ where $\vec{Z}$ is unimodular, and thus there is an infinite, but countable, number of vertices in the Ryshkov polytope. This observation also implies that the vertices can be partitioned into equivalence classes, where the equivalence relation is an isometry between two vertices. 

Lattices corresponding to vertices of $\mathcal{R}_{\lambda}$ are named \emph{perfect lattices} in the literature \cite{Martinet03}, and the corresponding Gram matrices are \emph{perfect forms}. The next theorem gives another interesting property of the Ryshkov polytope \cite{Martinet03}, which is important for our work. 
\vspace{0.2cm}
\begin{theorem}
\label{ryshkovthm2}
There are only finitely many non-isometric perfect forms in the Ryshkov polytope. 
\end{theorem} 
\vspace{0.2cm}
Hence, although there are infinitely many perfect forms in the Ryshkov polytope, Theorem \ref{ryshkovthm2} reveals that out of these, only finitely many are non-isometric and correspond to different lattices. 
The non-isometric perfect lattices have been tabulated for all dimensions up to $N = 8$ \cite{Achill}. In two and three dimensions, there is only one unique perfect lattice. In four dimensions, there are two, in five there are three, and in dimension $N = 8$ there are 10916.

\emph{Voronoi's algorithm}, is commonly used to traverse the vertices \cite{Martinet03} of the Ryshkov polytope. In Section \ref{optZsec} we present an algorithm for solving our problem, which in essence is a modified version of Voronoi's algoritm.

\subsection{Minkowski Polytope}
\label{minkredsec}
Another characterization of PQFs is via \emph{Minkowski reduction}. 
\begin{definition}
\label{minkreddef}
The Minkowski reduction region $\mathcal{M}$ is the set of all $\vec{G}$ satisfying
\bea
\label{minkred}
(i) & & \vec{v}^{\transp}\vec{G}\vec{v} \geq g_{i,i}, \quad \textnormal{for all $\vec{v} \in \mathbb{Z}^N$ such that $\mathrm{gcd}(v_i,\ldots,v_N) = 1.$} \nonumber \\
(ii) & & g_{i,i+1} \geq 0, \quad i = 1,\ldots,N-1.
\eea 
\end{definition}
As with the Ryshkov polytope, $\mathcal{M}$ is a subset of $\mathcal{S}_{\succ 0}^{N\times N}$.
A PQF $\vec{G} = \vec{L}^{\transp}\vec{L} \in \mathcal{M}$ is said to be \emph{Minkowski reduced} and the lattice generator matrix $\vec{L}$ is called a \emph{Minkowski reduced} generator matrix for $\Lambda_{\vec{L}}$. It can be shown that any lattice $\Lambda_{\vec{B}}$ has a generator matrix $\vec{L}$ that is Minkowski reduced, i.e., there exists an $\vec{L}$ such that $\vec{B} = \vec{LZ}$, where $\vec{L}$ is Minkowski reduced and $\vec{Z}$ is a unimodular matrix \cite{Min05}. Note that the Minkowski reduced generator matrix is not unique for a certain lattice, e.g., if $\vec{L}$ is Minkowski reduced, then so is $-\vec{L}$. However, it can be proved that there are only finitely many Minkowski reduced generator matrices for any lattice \cite{Min05}. Given a generator matrix $\vec{B}$, a Minkowski reduced generator matrix $\vec{L}$, and the corresponding unimodular matrix $\vec{Z}$, can both be obtained by applying the Minkowski reduction algorithm on $\vec{B}$ \cite{Min05}. 

Let $\vec{L}$ be a Minkowski reduced generator matrix and $\vec{G}_{\vec{L}} = \vec{L}^{\transp}\vec{L}$ the corresponding Minkowski reduced PQF. Condition (i) in \eqref{minkred} implies that $D^2_{\min}(\vec{L}) \geq g_{1,1}$, and since $g_{1,1} = \|\vec{l}_1\|^2$, it follows that $D^2_{\min}(\vec{L}) = g_{1,1}$, because at least $\vec{v} = (1\,0\,0\ldots 0)^{\transp}$ achieves equality. Hence, any Minkowski reduced generator matrix $\vec{L}$ contains the shortest vector in the lattice $\Lambda_{\vec{L}}$ as one of its columns. $\mathcal{M}$ is an intersection of infinitely many halfspaces, just as the Ryshkov polytope, but with different halfspaces in this case. It is easily seen that $\mathcal{M}$ corresponds to a cone in the vector space $\mathbb{R}^{N(N+1)/2}$, since if $\vec{G}_1, \vec{G}_2 \in \mathcal{M}$ , then $k_1\vec{G}_1 + k_2\vec{G}_2 \in \mathcal{M}$ for $k_1 \geq 0$ and $k_2 \geq 0$. We now define
\begin{definition}
\label{minkredlmbdef}
$\mathcal{M}_{\lambda} = \{\vec{G}\,\,:\,\, \vec{G} \in \mathcal{M},\,\, g_{1,1} \geq \lambda\}$.
\end{definition}
Hence, the Minkowski reduced PQFs in $\mathcal{M}_{\lambda}$  correspond to all Minkowski reduced generator matrices of lattices with a minimum distance of at least $\lambda$.

A polyhedral cone is a cone with a finite number of flat faces, and is therefore also a polytope. A fundamental result by Minkowski is \cite{Min05}
\begin{theorem}
\label{minkthm}
$\mathcal{M}$ is a polyhedral cone in $\mathbb{R}^{N(N+1)/2}$. 
\end{theorem}  
Theorem \ref{ryshkovthm1} is of later importance, since the polyhedral structure of $\vec{M}$ is crucial for the solvability of \eqref{prob2}.

It follows from Definition \ref{minkredlmbdef} that $\mathcal{M}_{\lambda}$ is the intersection of the hyperplane $\{\vec{G}\,:\,g_{1,1} = \lambda\}$ 
with $\mathcal{M}$, and from Theorem \ref{minkthm} we conclude that $\mathcal{M}_{\lambda}$ contains a finite number of vertices. The vertices are hereinafter denoted as \emph{Minkowski extreme forms}, and the corresponding lattices as \emph{Minkowski extreme lattices}. Minkowski extreme forms have been tabulated up to dimension 7, while for higher dimensions they are unknown since the computational complexity is tauntalizing. Compare this to perfect forms in the Ryshkov polytope, which are known up to dimension 8. This is due to the fact that enumerating perfect forms is computationally more tractable than enumerating Minkowski extreme forms \cite{Achill}. Ryshkov managed to show that every perfect form is equivalent to a form lying on an extreme ray of the Minkowski reduction region $\mathcal{M}$ \cite{Rys70}. Cohn et al., however, showed that there are extreme rays in the Minkowski reduction region that do not contain perfect forms \cite{CLR82}. Thus, every vertex (perfect form) of the Ryshkov polytope $\mathcal{R}_1$ can be reduced to a vertex in $\mathcal{M}_1$, but there are some extreme rays in $\mathcal{M}_1$ which contain PQFs in $\mathcal{R}_1$ that are not vertices of $\mathcal{R}_1$. Thus, since the Minkowski reduction region is different from the Ryshkov polytope, defining our optimization problem over it can provide additional insights to the properties of the optimal solution.

\subsection{Lattice-Based Problem Formulation}
\label{lattprobsec}
We are now ready to reformulate \eqref{prob2} as a pure lattice optimization problem. This will, for completeness, be done over the Ryshkov polytope as well as over the Minkowski reduction region. We begin with the former. Start by factorizing $\vec{F}$ as $\vec{F} = \vec{S}^{-1}\vec{UB}$, where $\vec{U}$ is an orthogonal matrix and $\vec{B}$ is any matrix such that $\vec{F}$ satisfies the trace constraint $\mathrm{tr}(\vec{F}^{\transp}\vec{F})\leq P_0$. From lattice theory, it follows that $\vec{B}$ can be regarded as a generator matrix for a lattice $\vec{\Lambda}_{\vec{B}}$. Inserting the expression $\vec{F} = \vec{S}^{-1}\vec{UB}$ into \eqref{prob2}, we arrive at
\be
\label{ryshproblatt}
\begin{gathered}
\min_{\vec{U},\vec{B}} \mathrm{tr}(\vec{B}^{\transp}\vec{U}^{\transp}\vec{S}^{-2}\vec{UB})\\
\textnormal{subject to}\\
\begin{aligned}
\vec{G}_{\vec{B}} \in \mathcal{R}_{1},
\end{aligned}
\end{gathered}
\ee
where $\vec{G}_{\vec{B}} = \vec{B}^{\transp}\vec{B}$. The subscript in $\vec{G}_{\vec{B}}$ will be left out when no confusion can arise. 

Let us now instead turn to the second formulation and formulate \eqref{prob2} as an optimization over the Minkowski polytope $\mathcal{M}_1$. We keep the factorization $\vec{F} = \vec{S}^{-1}\vec{UB}$, but we further factorize $\vec{B}$ as $\vec{B} = \vec{LZ}$, where $\vec{L}$ is a Minkowski reduced basis of the lattice $\Lambda_{\vec{B}}$ and $\vec{Z}$ a unimodular matrix. This gives that the $\vec{F}$ in \eqref{prob2} can also be factorized as $\vec{F} = \vec{S}^{-1}\vec{ULZ}$. Furthermore, the constraint $D^2_{\min}(\vec{B}) \geq 1$ is now equivalent to $\vec{G}_{\vec{L}} = \vec{L}^{\transp}\vec{L} \in \mathcal{M}_1$. Thus, \eqref{prob2} can as well be formulated as
\be
\label{minkproblatt}
\begin{gathered}
\min_{\vec{U},\vec{L},\vec{Z}} \mathrm{tr}(\vec{Z}^{\transp}\vec{L}^{\transp}\vec{U}^{\transp}\vec{S}^{-2}\vec{ULZ})\\
\textnormal{subject to}\\
\begin{aligned}
\vec{G}_{\vec{L}} \in \mathcal{M}_1.
\end{aligned}
\end{gathered}
\ee
Also for \eqref{minkproblatt}, the subscript in $\vec{G}_{\vec{L}}$ will sometimes be left out. Note that the objective functions in \eqref{ryshproblatt} and \eqref{minkproblatt} are exactly the same, since $\vec{B} = \vec{LZ}$, but the optimization procedure is different for the two problems. First of all, the optimization domains are different. Secondly,
the optimization in \eqref{ryshproblatt} only involves a minimization over orthogonal ($\vec{U}$) and invertible ($\vec{B}$) matrices, while \eqref{minkproblatt} is a minimization over orthogonal, invertible ($\vec{L}$) \emph{and} unimodular matrices ($\vec{Z}$). Hence, the methodology for solving \eqref{ryshproblatt} differs from the one solving \eqref{minkproblatt}. The formulation in \eqref{minkproblatt} also reveals the fact that changing the basis in $\Lambda_{\vec{B}}$, i.e., varying $\vec{Z}$, only affects the transmitted power, which is not evident from the formulation in \eqref{ryshproblatt}. Another advantage of the formulation in \eqref{minkproblatt} will be revealed by Theorem \ref{thm2} in Section \ref{minkredsec}. The work in \cite{Svante} considered a problem formulation similar to \eqref{minkproblatt}, but without using the Minkowski reduction domain. Instead, only the objective function was studied for different lattice bases $\vec{L}$. Both \cite{Forney} and \cite{Svante} made approximations to the minimum distance problem, the hypothesis was that the densest packing lattices in high dimensions should produce large distances. However, no exact results were presented. Instead, in \cite{Svante} it was just proposed that $\vec{L}$ should be a basis for the densest lattice packing, and a heuristic, iterative algorithm was given to find the optimal $\vec{Z}$. The derived precoders turn out to have good performance, however the question remains whether they indeed are optimal minimum distance precoders, and if not, how far away they are from the optimum. Thus, \eqref{minkproblatt} was not satisfactory treated in \cite{Svante}.

Two fundamental questions arise about the problems \eqref{ryshproblatt} and \eqref{minkproblatt}: 1) Is there an explicit formula for the optimal solutions to any of the problems? 2) If there is no such formula, what is the structure of the solution and can it be found in a simple way for any channel outcome $\vec{S}$? 

Before answering these questions, we look at some classical lattice problems and tools for solving lattice optimization problems. Our motivation for surveying these well known problems is that the problem studied in this paper is tightly connected, and has in principle the same structure in its solution as some of the classical problems.

\subsection{Classical Lattice Problems}
\label{classlatt}
There are many optimization problems that can be interpreted as optimization over lattices. A famous one is finding the \emph{densest lattice packing} of spheres in an $N$-dimensional space, corresponding to the following optimization
\be
\label{densprob}
\begin{gathered}
\min_{\vec{L}} \mathrm{Vol}(\vec{L}) \\
\textnormal{subject to}
\begin{aligned}
D^2_{\min}(\vec{L}) \geq 1,
\end{aligned}
\end{gathered}
\ee
i.e., to find, among all lattices with fixed minimum distance, the lattice with the minimal volume. 
The dual of this problem is to find the lattice maximizing the volume of the sphere \emph{encompassed} by its Voronoi region; this is known as maximizing the \emph{covering} of the lattice. Mathematically, it corresponds to the following optimization
\be
\label{coverprob}
\begin{gathered}
\min_{\vec{L}} \max_{\vec{w} \in \mathcal{V}(\vec{L})} \|\vec{w}\| \\
\textnormal{subject to} \\
\begin{aligned}
D^2_{\min}(\vec{L}) \geq 1.
\end{aligned}
\end{gathered}
\ee
The general solutions of these problems remain unknown as of today. However, for small enough dimensions, solutions are known. In two dimensions, it turns out that the hexagonal lattice solves both of these problems; this fact was shown for \eqref{densprob} by Lagrange in 1801 \cite{Lagrange}, and for \eqref{coverprob} by Kershner in 1934 \cite{CS88}. The packing problem has been solved for $N \leq 9$ and $N = 24$, while it is unsolved for all other $N$. For the covering problem, the solution is known for $N \leq 5$. Although the packing problem is unsolved in general, it is known that the optimal lattice must be a perfect lattice which also immediately implies that it is attained at a Minkowski extreme lattice. Hence, finding the densest lattice packing in any dimension $N$ amounts to traversing the non-isometric vertices in $\mathcal{R}_1$, or traversing the vertices in $\mathcal{M}_1$. Although the former is computationally more feasible, traversing the non-isometric perfect forms also becomes computationally inefficient for higher dimensions. Albeit the computational bottleneck, it is known that \eqref{densprob} and \eqref{coverprob} are both discrete optimization problems rather than continuous ones.

To show that the solution of \eqref{densprob} is achieved by a perfect lattice (or a Minkowski extreme lattice), it suffices to show that $\mathrm{Vol}(\vec{L})$ is a strictly concave function over $\mathcal{S}_{\succ 0}^{N\times N}$. Since $\mathcal{S}_{\succ 0}^{N\times N}$ contains the polytopes $\mathcal{R}_1$ and $\mathcal{M}_1$, this therefore implies that $\mathrm{Vol}(\vec{L})$ is concave over both $\mathcal{R}_1$ and $\mathcal{M}_1$. Therefore, the solution to \eqref{densprob} is attained at the vertices of these polytopes, i.e., at the perfect lattices (vertices of $\mathcal{R}_1$) and the Minkowski extreme lattices (vertices of $\mathcal{M}_1$). Hence, \emph{concavity} of the objective function is enough to conclude that perfect lattices solve a given lattice optimization problem. The concavity of $\det(\vec{G})^{1/N}$ over $\mathcal{S}_{\succ 0}^{N\times N}$ was shown by Minkowski \cite{Min05}. 

Next, we show that the objective functions in \eqref{ryshproblatt} and \eqref{densprob} are of different nature.
The orthogonal matrix $\vec{U}$ minimizing the objective function in \eqref{ryshproblatt} has been found in \cite{Svante}, and is shown to be equal to the left orthogonal matrix in the SVD decomposition of $\vec{B}$. Hence, inserting this $\vec{U}$ into the objective function gives the optimization
\be
\label{ourproblatt2}
\begin{gathered}
\min_{\vec{G}} \sum_{j=1}^N \omega_j(\vec{G}_{\vec{B}})/s_j^2 \\
\textnormal{subject to}
\begin{aligned}
\vec{G}_{\vec{B}} \in \mathcal{R}_1,
\end{aligned}
\end{gathered}
\ee
where $\omega_j(\vec{G}_{\vec{B}})$ is the $j$:th largest eigenvalue of $\vec{G}_{\vec{B}}$ and $s_j$ is the $j$:th largest diagonal element in $\vec{S}$. The optimization in \eqref{densprob} can be performed over the Ryshkov polytope, with the objective function $\sqrt[N]{\det(\vec{B}^{\transp}\vec{S}^{-2}\vec{B})} = \sqrt[N]{\det(\vec{S}^{-2})\det(\vec{G}_{\vec{B}})}$. The matrix $\vec{S}$ can be regarded as a constant and does not impact the optimization. It further holds that $$\sqrt[N]{\det(\vec{S}^{-2})\det(\vec{G}_{\vec{B}})} = \sqrt[N]{\prod_{j=1}^N \omega_j(\vec{G}_{\vec{B}})/s_j^2}.$$ Hence, \eqref{densprob} minimizes the $N$:th root of the product of the eigenvalues of $\vec{G}_{\vec{B}}$ over $\mathcal{R}_1$, while \eqref{ourproblatt2} minimizes a weighted sum of them. Due to the arithmetic-geometric mean (AM-GM) inequality, we have that $$\sum_{j=1}^N \omega_j(\vec{G})/s_j^2 \geq N\sqrt[N]{\prod_{j=1}^N\omega_j(\vec{G})/s_j^2},$$ which shows that the $\vec{L}$ solving \eqref{densprob} is only minimizing the lower bound to the objective function in \eqref{ourproblatt2}, thus not guaranteeing that it is the optimum to \eqref{ourproblatt2}\footnote{This same reasoning is used in \cite{Svante} in order to propose densest lattices as good candidates for providing a large minimum distance.} Hence, although \eqref{ryshproblatt} and \eqref{densprob} have the same optimization domain, \eqref{ryshproblatt} posseses a different objective function than \eqref{densprob}, and is thus a different lattice optimization problem. 

\section{Optimal Lattice Structure}
\label{probsolsec}
This section will prove the concavity of the objective functions in \eqref{ryshproblatt} and \eqref{minkproblatt}, respectively. 
We start by proving the concavity of the objective function in \eqref{minkproblatt} over $\mathcal{S}_{\succ 0}^{N\times N}$, for any given $\vec{S}$ and $\vec{Z}$ matrix. Define $f(\vec{L},\vec{Z}) \stackrel{\triangle}{=} \min_{\vec{U}} \mathrm{tr}(\vec{Z}^{\transp}\vec{L}^{\transp}\vec{U}^{\transp}\vec{S}^{-2}\vec{ULZ})$, which is the objective function in \eqref{minkproblatt} without the minimization over $\vec{L}$ and $\vec{Z}$. Observe also that $f(\vec{B},\vec{I}_N)$ is the objective function in \eqref{ryshproblatt} without the minimization over $\vec{B}$. We now show
\begin{theorem}
\label{thm2}
For a fixed $\vec{Z}$, $f(\vec{L},\vec{Z})$ is concave over $\mathcal{S}_{\succ 0}^{N\times N}$ with respect to $\vec{G}_{\vec{L}}$.
\end{theorem}
\begin{proof}
Write $\min_{\vec{U}} \mathrm{tr}(\vec{Z}^{\transp}\vec{L}^{\transp}\vec{U}^{\transp}\vec{S}^{-2}\vec{ULZ}) = \min_{\vec{U}} \mathrm{tr}(\vec{S}^{-2}\vec{ULZ}\vec{Z}^{\transp}\vec{L}^{\transp}\vec{U}^{\transp})$. Let $\vec{LZ}\vec{Z}^{\transp}\vec{L}^{\transp} = \vec{QD}\vec{Q}^{\transp}$ be the eigenvalue decomposition of $\vec{LZ}\vec{Z}^{\transp}\vec{L}^{\transp}$, where $\vec{Q}$ is the orthogonal matrix. Now note that 
\begin{eqnarray*}
f(\vec{L},\vec{Z}) & = & \min_{\vec{U}} \mathrm{tr}(\vec{S}^{-2}\vec{ULZ}\vec{Z}^{\transp}\vec{L}^{\transp}\vec{U}^{\transp}) \\
& = & \min_{\vec{U}} \mathrm{tr}(\vec{S}^{-2}\vec{UQD^2}\vec{Q}^{\transp}\vec{U}^{\transp}) \\
& = & \min_{\vec{U}} \mathrm{tr}(\vec{S}^{-2}\vec{U}\vec{Q}^{\transp}\vec{D}^2\vec{Q}\vec{U}^{\transp}) \\
& = & \min_{\vec{U}} \mathrm{tr}(\vec{S}^{-2}\vec{U}\vec{Z}^{\transp}\vec{L}^{\transp}\vec{LZ}\vec{U}^{\transp}) \\
& = & \min_{\vec{U}} \mathrm{tr}(\vec{S}^{-2}\vec{U}\vec{Z}^{\transp}\vec{G}_{\vec{L}}\vec{Z}\vec{U}^{\transp}).  
\end{eqnarray*}
Hence 
\begin{eqnarray*}
f(\vec{L},\vec{Z}) & = & h(\vec{G}_{\vec{L}},\vec{Z}) \\
& = & \min_{\vec{U}} \mathrm{tr}(\vec{S}^{-2}\vec{U}\vec{Z}^{\transp}\vec{G}_{\vec{L}}\vec{Z}\vec{U}^{\transp}).
\end{eqnarray*}
Now it follows that for positive semidefinite $\vec{G}_1$, $\vec{G}_2$ and $0 \leq \gamma \leq 1$, 
\begin{eqnarray*}
h(\gamma \vec{G}_1 + (1-\gamma)\vec{G}_2,\vec{Z}) & = & \min_{\vec{U}}\gamma\mathrm{tr}(\vec{S}^{-2}\vec{U}\vec{G}_1\vec{U}^{\transp}) + (1-\gamma)\mathrm{tr}(\vec{S}^{-2}\vec{U}\vec{G}_2\vec{U}^{\transp}) \nonumber \\
& \geq & \gamma\min_{\vec{U}}\mathrm{tr}(\vec{S}^{-2}\vec{U}\vec{G}_1\vec{U}^{\transp}) + (1-\gamma)\min_{\vec{U}}\mathrm{tr}(\vec{S}^{-2}\vec{U}\vec{G}_2\vec{U}^{\transp}) \nonumber \\
& = & \gamma h(\vec{G}_1) + (1-\gamma)h(\vec{G}_2), 
\end{eqnarray*}
which shows that $h(\vec{G}_{\vec{L}},\vec{Z})$, and thus also $f(\vec{L},\vec{Z})$, are concave over $\mathcal{S}_{\succ 0}^{N\times N}$ with respect to $\vec{G}_{\vec{L}}$. 
\end{proof}
An immediate corollary is
\begin{corollary}
\label{corr1}
$f(\vec{L},\vec{Z})$ is concave over $\mathcal{S}_{\succ 0}^{N\times N}$ with respect to the Gram matrix $\vec{G}_{\vec{B}} = \vec{G}_{\vec{LZ}} = \vec{Z}^{\transp}\vec{L}^{\transp}\vec{LZ}$. 
\end{corollary}
\begin{proof}
Let $\vec{L} = \vec{B}$, $\vec{Z} = \vec{I}_N$ and apply Theorem \ref{thm2}.
\end{proof}
Taken together, Theorem \ref{thm2} and Corollary 1 show that the objective functions in \eqref{ryshproblatt} and \eqref{minkproblatt}, respectively, are both concave over their corresponding domains. This immediately implies that the solution to \eqref{ryshproblatt} is a perfect lattice, and the solution to \eqref{minkproblatt} is a Minkowski extreme lattice. Exactly which perfect lattice/Minkowski extreme lattice that solves \eqref{prob2} depends of course on the channel outcome $\vec{S}$, and an algorithm is given in Section \ref{optZsec} that enumerates all possible perfect forms solving \eqref{prob2} for a specific $\vec{S}$. However, since there are finitely many perfect lattices/Minkowski extreme lattices in $N$ dimensions, we know that there are finitely many different lattices solving the problem for all $\vec{S}$. This answers our second question posed in Section \ref{lattprobsec}: The optimal $\vec{G}_{\vec{B}}$ in \eqref{ryshproblatt} is a vertex of the polytope $\mathcal{R}_1$, and the optimal $\vec{G}_{\vec{L}}$ solving \eqref{minkproblatt} corresponds to a vertex in $\mathcal{M}_1$. Hence, the solution to \eqref{prob2} does not depend continuously on $\vec{S}$, instead it changes in a discrete fashion when $\vec{S}$ is varied continuously. Relating to our first question in Section \ref{lattprobsec}, this result implies that an explicit formula for the solution of \eqref{prob2} seems out of reach, since such a formula does not exist for \eqref{densprob} whose set of possible solutions is a subset of the set of possible solutions to \eqref{prob2}. Altogether, a previously unknowned result is revealed: There are finitely many lattices that can solve the minimum distance optimization problem in \eqref{prob2}, and they can be enumerated offline.

Theorem \ref{thm2} also reveals that for \emph{any} given $\vec{Z}$ matrix, the optimal solution to \eqref{minkproblatt} occurs at a Minkowski extreme lattice. Thus, given any $\vec{Z}$, the optimal $\vec{L}$ that builds up $\vec{F} = \vec{S}^{-1}\vec{ULZ}$ in \eqref{prob2} is a Minkowski extreme lattice. This fact will be used in Section \ref{optZsec} to develop a good suboptimal precoder construction. Hence, the problem formulation in \eqref{minkproblatt} provides additional information about the behavior of \eqref{prob2}, not present in \eqref{ryshproblatt}: This is the main reason for introducing \eqref{minkproblatt}.

We can already at this stage deduce several interesting conclusions from the result in Theorem \ref{thm2}. For up to three dimensions, there is only one non-isometric perfect lattice in each dimension: For $N = 2$ it is the hexagonal lattice and in $N = 3$ it is the face-centered cubic lattice. Since these are the only non-isometric perfect lattices in these dimensions, they also solve \eqref{densprob}, and thus the proposition in \cite{Svante} to use densest lattice packings in \eqref{minkproblatt} is optimal for these dimensions. However, when $N = 4$ there are two non-isometric lattices: The checkerboard lattice $D_4$ and the root lattice $A_4$ \cite{CS88}. It will be demonstrated in Section \ref{denscontrsec} that \emph{both} of these lattices occur as solutions to \eqref{ryshproblatt} for different $\vec{S}$, so the constructions in \cite{Svante} are suboptimal for $N = 4$. 
Another interesting consequence of Theorem \ref{thm2} is that the main result in \cite{KCMR11}, which shows that the hexagonal lattice is optimal in two dimensions, now follows immediately from Theorem \ref{thm2}.

Note that it is now an easy task to construct the optimal $\vec{F}$ in \eqref{mod3}, once the $\vec{G}_{\vec{B}}$ solving \eqref{ryshproblatt} is known. Let $\vec{G}_{\mathrm{opt}} = \vec{B}_{\mathrm{opt}}^{\transp}\vec{B}_{\mathrm{opt}}$ denote the optimal form and $\vec{G}_{\mathrm{opt}} = \vec{Q}_{\mathrm{opt}}\vec{D}_{\mathrm{opt}}\vec{Q}_{\mathrm{opt}}^{\transp}$ be its eigenvalue decomposition. Since the optimal $\vec{U}$ in $\vec{F} = \vec{S}^{-1}\vec{UB}$ is equal to the left orthogonal matrix in the SVD decomposition of $\vec{B}$, it follows that the optimal $\vec{F}$ can be constructed as 
\be
\label{optFconstr}
\vec{F}_{\mathrm{opt}} = \vec{S}^{-1}\sqrt{\vec{D}_{\mathrm{opt}}}\vec{U}_{\mathrm{opt}}^{\transp}.
\ee

To summarize, we now have the following knowledge at hand about the solution to our original problem in \eqref{prob2}. We have shown that \eqref{prob2} is equivalent to both \eqref{ryshproblatt} and \eqref{minkproblatt}. Theorem \ref{thm2} then shows that the $\vec{L}$ matrix solving \eqref{minkproblatt}, for \emph{any} invertible $\vec{S}$, gives rise to a Gram matrix $\vec{G} = \vec{L}^{\transp}\vec{L}$ that corresponds to a vertex in the polytope $\mathcal{M}_1$. Since $\mathcal{M}_1$ has a finite number of vertices for any dimension $N$, and is independent of the matrix $\vec{S}$, it holds that there are finitely many $\vec{L}$ matrices (up to rotation) that are candidates to solving \eqref{minkproblatt} for any given $\vec{S}$. Once the optimal $\vec{L}_{\mathrm{opt}}$ is known (up to rotation), it remains to find the optimal unimodular matrix $\vec{Z}_{\mathrm{opt}}$ in \eqref{minkproblatt} and then to construct the optimal precoder $\vec{F}_{\mathrm{opt}}$ from \eqref{optFconstr}, where $\vec{B}_{\mathrm{opt}} = \vec{L}_{\mathrm{opt}}\vec{Z}_{\mathrm{opt}}$. To find the optimal $\vec{Z}$, we need to perform a search over unimodular matrices, which can be simplified if good bounds on the optimum solution to \eqref{minkproblatt} are known. These bounds will be developed in the next section.

When it comes to the equivalent problem formulation in \eqref{ryshproblatt}, Corollary \ref{corr1} shows that the $\vec{B}$ matrix solving \eqref{ryshproblatt} for any given invertible $\vec{S}$, is such that it produces a Gram matrix $\vec{G} = \vec{B}^{\transp}\vec{B}$ that is one of the vertices in the polytope $\mathcal{R}_1$. Given the optimal $\vec{G}$ in $\mathcal{R}_1$, the optimal precoder is obtained through \eqref{optFconstr}. The $\mathcal{R}_1$ polytope contains infinitely many vertices, and it is known that each vertex is isometric to some vertex in $\mathcal{M}_1$. Since \eqref{ryshproblatt} is connected to \eqref{minkproblatt} through the factorization $\vec{B} = \vec{LZ}$, it holds that if $\vec{G}_{\vec{B}} = \vec{B}^{\transp}\vec{B} = \vec{Z}^{\transp}\vec{L}^{\transp}\vec{LZ}$ is a vertex in $\mathcal{R}_1$, then $\vec{G}_{\vec{L}} = \vec{L}^{\transp}\vec{L}$ is a vertex in $\mathcal{M}_1$. Hence, traversing the different vertices in $\mathcal{R}_1$ is equivalent to a joint enumeration of some of the vertices in $\mathcal{M}_1$ and different unimodular matrices $\vec{Z}$. However, it is instead possible to directly enumerate perfect forms by formulating an algorithm working over $\mathcal{R}_1$. Again, bounds are needed in order to restrict the amount of vertices to traverse, and they will be presented in the next section.


\section{Bounds On the Optimal Solution}
\label{boundssec}
We start by deriving lower and upper bounds to $\mathrm{tr}(\vec{Z}^{\transp}\vec{L}^{\transp}\vec{U}^{\transp}\vec{S}^{-2}\vec{ULZ})$, which is the objective function in \eqref{minkproblatt}. The upper bound presented here improves significantly upon the upper bound presented in \cite{Svante}. From these bounds, we are able to derive further bounds that aid in restricting the search space for the algorithms to be introduced in Section \ref{optZsec}.

\begin{theorem}
\label{lbthm}
The following lower bound holds for the optimal solution to \eqref{minkproblatt}
\be
\label{lb}
\min_{\vec{U},\vec{L},\vec{Z}}\mathrm{tr}(\vec{Z}^{\transp}\vec{L}^{\transp}\vec{U}^{\transp}\vec{S}^{-2}\vec{UL}\vec{Z}) \geq N\sqrt[N/2]{\det(\vec{L})/\det(\vec{S})}. 
\ee
\end{theorem}
\begin{proof}
Dropping the integer-valued constraint on $\vec{Z}$, while keeping the determinant constraint $\det(\vec{Z}) = \pm 1$, we apply the method of Lagrange multipliers to find first order optimality conditions. Let $\vec{M} = \vec{L}^{\transp}\vec{U}^{\transp}\vec{S}^{-2}\vec{UL}$. The optimal $\vec{Z}_o$ must satisfy
\be
\label{firstordopt}
\frac{\partial \mathrm{tr}(\vec{Z}_o^{\transp}\vec{L}^{\transp}\vec{U}^{\transp}\vec{S}^{-2}\vec{UL}\vec{Z}_o)}{\partial \vec{Z}_o} = \gamma\frac{\det(\vec{Z}_o)-1}{\partial \vec{Z}_o} \Rightarrow 2\vec{Z}_o^{\transp}\vec{M} = \gamma\left(\vec{Z}_o^{\transp}\right)^{-1},
\ee
where $\gamma \in \mathbb{R}$. Taking determinants on both sides, and making use of $\det(\vec{Z}_o) = \pm 1$, we get $\gamma = 2\sqrt[N]{\det(\vec{M})}$. Inserting this $\gamma$ into \eqref{firstordopt} and multiplying both sides of the equation with $\vec{Z}_o^{\transp}$, we arrive at $\vec{Z}_o^{\transp}\vec{M}\vec{Z}_o = \sqrt[N]{\det(\vec{M})}\vec{I}_N$. Hence, for this $\vec{Z}_o$, we get 
\begin{eqnarray*}
\mathrm{tr}(\vec{Z}_o^{\transp}\vec{L}^{\transp}\vec{U}^{\transp}\vec{S}^{-2}\vec{UL}\vec{Z}_o) & = & \mathrm{tr}(\vec{I}_N)\sqrt[N]{\det(\vec{M})} \\  
& = & N\sqrt[N]{\det(\vec{M})}.
\end{eqnarray*} 
Since this expression is independent of $\vec{U}$, it is a lower bound to the objective function for a fixed $\vec{L}$. Expressing $\det(\vec{M}) = \det(\vec{L}^{\transp}\vec{S}^{-2}\vec{L}) = \det^2(\vec{L})/\det^2(\vec{S})$, we arrive at the lower bound in \eqref{lb}.
\end{proof}
This bound was also reported in \cite{Svante}, but derived in a different way, by using the AM-GM inequality and Hadamard's inequality. The approach presented here shows that this lower bound corresponds to the optimal \emph{real-valued} unimodular $\vec{Z}$. 

Next, we derive an upper bound on $\min_{\vec{U},\vec{L},\vec{Z}}\mathrm{tr}(\vec{Z}^{\transp}\vec{L}^{\transp}\vec{U}^{\transp}\vec{S}^{-2}\vec{UL}\vec{Z})$

\begin{theorem}
\label{ubthm}
If $D^2_{\min}(\vec{L}) = 1$, then
\be
\label{ub}
\min_{\vec{U},\vec{L},\vec{Z}}\mathrm{tr}(\vec{Z}^{\transp}\vec{L}^{\transp}\vec{U}^{\transp}\vec{S}^{-2}\vec{UL}\vec{Z}) \leq N\sqrt[N/2]{1/\det(\vec{S})}.
\ee
\end{theorem}
\begin{proof}
Let $\vec{B} = \vec{ULZ} = \vec{QR}$ denote the QR-decomposition of the received lattice. It holds that $D^2_{\min}(\vec{B}) \geq \min_{1 \leq i \leq N} |r_{i,i}|^2$ \cite{Mow2006}. In \cite{Hager}, an orthogonal precoder matrix $\vec{F}_{\mathrm{gmd}}$ (geometric mean precoder) and an orthogonal receiver matrix $\vec{W}_{\mathrm{gmd}}$ were constructed, such that in the QR decomposition of $\vec{B}_{\mathrm{gmd}} = \vec{W}_{\mathrm{gmd}}\vec{S}\vec{F}_{\mathrm{gmd}}$, all diagonal elements of $\vec{R}$ equal $\sqrt[N]{\det(S)}$. Hence, in essence, the precoder $\vec{F}_{\mathrm{gmd}}$ together with the rotation $\vec{W}_{\mathrm{gmd}}$ at the receiver, produces a lattice with maximal lower bound on $D^2_{\min}$. This value is equal to the geometric mean of its singular values, thereby its name the ''geometric mean precoder''. It is clear that $D^2_{\min}(\vec{B}_{\mathrm{gmd}}) \geq \sqrt[N]{\det(S)}$, and since $\vec{F}_{\mathrm{gmd}}$ is orthogonal, $\mathrm{tr}(\vec{F}^{\transp}_{\mathrm{gmd}}\vec{F}_{\mathrm{gmd}}) = N$. Hence $d^2_{\min}(\vec{S},\vec{F}_{\mathrm{gmd}}) = D^2_{\min}(\vec{B}_{\mathrm{gmd}})/\mathrm{tr}(\vec{F}^{\transp}_{\mathrm{gmd}}\vec{F}_{\mathrm{gmd}}) \geq \sqrt[N/2]{\det(\vec{S})}/N$. Now it follows that for any precoder $\vec{F}$ with higher $d^2_{\min}(\vec{S},\vec{F})$ than $d^2_{\min}(\vec{S},\vec{F}_{\mathrm{gmd}})$, $d^2_{\min}(\vec{S},\vec{F}) \geq (\det(\vec{S}))^{2/N}/N$. Hence, this gives an upper bound on $\mathrm{tr}(\vec{F}\vec{F}^{\transp})$, $\mathrm{tr}(\vec{F}\vec{F}^{\transp}) \leq N\sqrt[N/2]{1/\det(\vec{S})}$. Writing $\vec{F} = \vec{S}^{-1}\vec{ULZ}$, we get the upper bound in \eqref{ub}.
\end{proof}

Combining Theorem \ref{lbthm} and \ref{ubthm}, we have the following bounds
\be
\label{enbounds}
N(\det(\vec{L})/\det(\vec{S}))^{2/N} \leq \mathrm{tr}(\vec{Z}^{\transp}\vec{L}^{\transp}\vec{U}^{\transp}\vec{S}^{-2}\vec{U}\vec{L}\vec{Z}) \leq N(1/\det(\vec{S}))^{2/N}.
\ee
\eqref{enbounds} translates into 
\be
\label{dminbounds}
(\det(\vec{S})/\det(\vec{L}))^{2/N}/N \geq d_{\min}^2(\vec{S},\vec{F}) \geq (\det(\vec{S}))^{2/N}/N.
\ee
Note that the $\vec{L}$ in \eqref{enbounds} and \eqref{dminbounds} is such that $D^2_{\min}(\vec{L}) = 1$.
Also, these bounds hold for \emph{any} precoder $\vec{F} = \vec{S}^{-1}\vec{ULZ}$ that improves upon $\vec{F}_{\mathrm{gmd}}$.
It is readily seen that the ratio between the upper bound and the lower bound in \eqref{enbounds} is $(1/\det(\vec{L}))^{2/N}$. This can be compared to the bounds in \cite{Svante}, where the ratio between the upper bound and lower bound contains the exponential factor $2^{N/2}$. Thus, the improvement in the upper bound is significant. Hence, for a fixed dimension, the optimum ratio $d^2_{\min}(\vec{S}\vec{F}_{\mathrm{opt}})/(\det(\vec{S}))^{2/N}$ is always smaller than $(1/\det(\vec{L}))^{2/N}$, \emph{independently} of the channel $\vec{S}$. For example, when $N = 2$, the optimal lattice is the hexagonal lattice $\vec{L}_{\mathrm{hex}}$ and when $D^2_{\min}(\vec{L}_{\mathrm{hex}}) = 1$, $\det(\vec{L}_{\mathrm{hex}}) = 1/2$. The ratio between the upper bound and lower bound in \eqref{dminbounds} is then 2, hence the optimal $d^2_{\min}(\vec{S}\vec{F}_{\mathrm{opt}})$ in two dimensions is at most twice the lower bound $d_{\min}^2(\vec{S}\vec{F}_{\mathrm{gmd}}) = \det(\vec{S})/2$. For $N = 3$, the optimal lattice is the face-centered cubic lattice $\vec{L}_{A_3}$ that has a volume of $\det(\vec{L}_{A_3}) = 1/2$ when $D^2_{\min}(\vec{L}_{A_3}) = 1$. In this case, the ratio between the bounds in \eqref{dminbounds} is $2^{2/3} \approx 1.59$; hence, the performance of the optimal precoder is at most $59\%$ better than for $\vec{F}_{\mathrm{gmd}}$. It is worthwile to observe that the ratio between the bounds, for optimal packing lattices $\vec{L}$, equals Hermite's constant $\delta_N \stackrel{\triangle}{=} \max_{\vec{L}} D^2_{\min}(\vec{L})/\mathrm{Vol}(\vec{L})^{2/N}$. Hermite's constant is the ratio between the constraint function and the objective function in \eqref{densprob}, and is therefore an optimization problem equivalent to \eqref{densprob}. The following upper and lower bounds are known for $\delta_N$ \cite{CS88,Milnor} $$ \frac{N}{2\pi e} + \frac{\log(\pi N)}{2\pi e} + c_{N,1} \leq \delta_N \leq \frac{1.744N}{2\pi e}(1 + c_{N,2}),$$
where $c_{N,1}$ and $c_{N,2}$ are constants depending on the dimension $N$. From this it follows that $\delta_N$ grows linearly with the dimension $N$, and thus the ratio of our bounds grows at most linearly with $N$. Note however that $\vec{F}_{\mathrm{gmd}}$ operates above the lower bound in \eqref{dminbounds}, and due to its good SER performance as reported in \cite{Hager}, it can serve as a basis for developing a suboptimal $\vec{Z}$ matrix to \eqref{minkproblatt}. This will be presented in Section \ref{optZsec}.

As discussed in the first paragraph of Section \ref{probsolsec}, a closed form solution to \eqref{prob2} seems out of reach. We are thus interested in an algorithm that can find the optimal $\vec{Z}$ in \eqref{minkproblatt}, or an algorithm to find the optimal $\vec{G}_{\vec{B}}$ in \eqref{ryshproblatt}. In order to do so, it is desirable to first have some bounds on the $\vec{Z}$ matrix or some quantity depending on it, in order to restrict the search space.
From this perspective, we develop an upper bound on $\mathrm{tr}(\vec{Z}^{\transp}\vec{L}^{\transp}\vec{LZ})$.
\begin{theorem}
\label{Zubthm}
With $D^2_{\min}(\vec{L}) = 1$, the following upper bound holds
\be
\label{enub}
\min_{\vec{L},\vec{Z}}\mathrm{tr}(\vec{Z}^{\transp}\vec{L}^{\transp}\vec{L}\vec{Z}) \leq N\left(\frac{s_1}{\sqrt[n]{\det(\vec{S})}}\right)^2.
\ee
\end{theorem}
\begin{proof}
Let $\vec{G}_{\vec{L}} = \vec{Z}^{\transp}\vec{L}^{\transp}\vec{L}\vec{Z}$. Inserting the optimal $\vec{U}$ into \eqref{enbounds}, we arrive at the upper bound
\be
\label{enbounds2}
\mathrm{tr}(\Omega(\vec{G}_{\vec{L}})\vec{S}^{-2}) \leq N(1/\det(\vec{S}))^{2/N},
\ee
where $\Omega(\vec{G}_{\vec{L}})$ is the diagonal matrix containing the eigenvalues $\omega_j(\vec{G}_{\vec{L}})$ of $\vec{G}_{\vec{L}}$. Since the eigenvalues $\omega_j(\vec{G}_{\vec{L}})$ are sorted in opposite order to $s_j^{-2}$, and $s_1^{-2} \leq \ldots \leq s_N^{-2}$, we have the inequality $$\mathrm{tr}(\Omega(\vec{G}_{\vec{L}})\vec{S}^{-2}) \geq \frac{\mathrm{tr}(\Omega(\vec{G}_{\vec{L}}))}{s_1^2},$$ which gives us the upper bound in \eqref{enub}
\end{proof}
Geometrically, the inequality in \eqref{enub} implies that the lattice vectors of the optimal lattice $\vec{L}_{\mathrm{opt}}\vec{Z}_{\mathrm{opt}}$ must have a bounded length. Using the trace inequality \cite{KA68} $$\omega_N(\vec{G}_{\vec{L}_{\mathrm{opt}}})\mathrm{tr}(\vec{Z}_{\mathrm{opt}}\vec{Z}_{\mathrm{opt}}^{\transp}) \leq \mathrm{tr}(\vec{Z}_{\mathrm{opt}}^{\transp}\vec{L}_{\mathrm{opt}}^{\transp}\vec{L}_{\mathrm{opt}}\vec{Z}_{\mathrm{opt}}) \leq \omega_1(\vec{G}_{\vec{L}_{\mathrm{opt}}})\mathrm{tr}(\vec{Z}_{\mathrm{opt}}\vec{Z}_{\mathrm{opt}}^{\transp}),$$
we also have the following upper bound for $\vec{Z}_{\mathrm{opt}}$
\be
\label{optZbounds}
\mathrm{tr}(\vec{Z}_{\mathrm{opt}}\vec{Z}_{\mathrm{opt}}^{\transp}) \leq \frac{N}{\omega_N(\vec{G}_{\vec{L}_{\mathrm{opt}}})}\left(\frac{s_1}{\sqrt[n]{\det(\vec{S})}}\right)^2.
\ee

In terms of $\vec{G}_{\vec{B}} = \vec{Z}^{\transp}\vec{L}^{\transp}\vec{LZ} = \vec{B}^{\transp}\vec{B}$, the upper bound in \eqref{enub} is
\be
\label{enub2}
\mathrm{tr}(\vec{G}_{\vec{B}}) \leq N\left(\frac{s_1}{\sqrt[n]{\det(\vec{S})}}\right)^2.
\ee
Let $\mathrm{ub}(\vec{S})$ denote the upper bound in \eqref{enub2}. Hence, the optimal $\vec{G}_{\vec{B}}$ in the Ryshkov polytope that solves \eqref{ryshproblatt} is one of the vertices of the finite, bounded polytope $\mathcal{R}_1 \cap \{\vec{G}_{\vec{B}}\,:\,\mathrm{tr}(\vec{G}_{\vec{B}}) \leq \mathrm{ub}(\vec{S})\}$.

\section{Numerical Methods}
\label{optZsec}
In this section, we present algorithmical approaches to solve \eqref{ryshproblatt} and \eqref{minkproblatt}, along with a novel suboptimal precoder construction based on Theorem \ref{thm2}. Beside this, by applying the knowledge that perfect forms solve \eqref{ryshproblatt}, we provide a numerical example where the densest lattice packing is not a solution to \eqref{ryshproblatt}. This shows that the solution to \eqref{prob2} is somewhat counter-intuitive: The optimal packing of points at the receiver, does not always minimize the total energy of the lattice points at the transmitter.

\subsection{Methods to solve \eqref{ryshproblatt} and \eqref{minkproblatt}}
Although \eqref{ryshproblatt} can be solved by enumerating all vertices in the polytope $\mathcal{R}_1 \cap \{\vec{G}\,:\,\mathrm{tr}(\vec{G}) \leq \mathrm{ub}(\vec{S})\}$, the methodology for solving \eqref{minkproblatt} will provide another interesting observation. Additionally, the problem formulation in \eqref{minkproblatt} gives novel insight into efficient suboptimal precoder construction, not present in the formulation in \eqref{ryshproblatt}. We divide this section into two parts: First we discuss a method to solve \eqref{minkproblatt}, then we discuss the solution to \eqref{ryshproblatt}. 
 
\subsubsection{Finding the solution to \eqref{minkproblatt}}
\label{minkprobsolsec}
To find the solution to \eqref{minkproblatt}, one needs to tabulate Minkowski extreme lattices in $N$ dimensions. Unfortunately, it turns out to be more complex to enumerate Minkowski extreme lattices than perfect forms \cite{Achill}. Note, however, that only those Minkowski extreme lattices that correspond to perfect forms have to be known. Namely, once all non-isometric perfect forms have been tabulated in $N$ dimensions, the fact that each perfect form (vertex) in $\mathcal{R}_1$ is equivalent to a Minkowski extreme form (vertex) in $\mathcal{M}_1$, implies that the Minkowski extreme lattices solving \eqref{minkproblatt} are the ones corresponding to the non-isometric perfect forms. Therefore, it is not necessary to know \emph{all} the Minkowski extreme lattices in $N$ dimensions in order to solve \eqref{minkproblatt}, only those corresponding to perfect forms are needed. However, when constructing a good suboptimal solution to \eqref{minkproblatt}, presented in Section \ref{subprecsec}, it is necessary to know all the Minkowski extreme lattices to obtain the best suboptimal construction. 
The smallest eigenvalue $\omega_N(\vec{G}_{\vec{L}})$ in \eqref{optZbounds} is non-zero for all the Minkowski extreme lattices $\vec{L}$ (candidates for the optimum), and thus the bound in \eqref{optZbounds} is well-defined. A geometrical interpretation is that this inequality bounds the squared lengths sum of the basis vectors in the integer lattice $\mathbb{Z}^{N}$, where the basis vectors are now the rows of $\vec{Z}$. Thus, finding the optimal $\vec{Z}$ can be regarded as searching for basis vectors inside a sphere of a certain radius. 
If one has a priori knowledge about the maximum ratio $s_1/\det(\vec{S})$, an off-line, one-shot algorithm can be formulated that searches for unimodular $\vec{Z}$ inside the largest sphere, corresponding to the Minkowski extreme lattice $\vec{L}$ with smallest $\omega_N(\vec{G}_{\vec{L}})$ and the channel $\vec{S}$ with largest upper bound in \eqref{optZbounds}. This sphere certainly includes the optimal $\vec{Z}_{\mathrm{opt}}$ corresponding to the optimal Minkowski extreme lattice $\vec{L}_{\mathrm{opt}}$ for any channel that can occur. A large codebook of matrices $\vec{LZ}$ can then be constructed off-line, by multiplying each encountered $\vec{Z}$ in the sphere with the different Minkowski extreme lattices and storing the resulting matrices into the codebook. To then find the optimal precoder $\vec{F} = \vec{S}^{-1}\vec{ULZ}$ online for a certain $\vec{S}$, one simply goes through every element $\vec{LZ}$ in the codebook, and constructs $\vec{F}$ by using the optimal $\vec{U}$. 

The outlined method to solve \eqref{minkproblatt} includes the following steps: 1) Find all Minkowski extreme lattices that correspond to non-isometric perfect forms. This is accomplished by applying Voronoi's algorithm to enumerate non-isometric perfect forms \cite{Voronoi07}, and then applying the Minkowski reduction algorithm to the obtained perfect lattices. 2) Enumerate all unimodular $\vec{Z}$ satisfying the bounds in \eqref{optZbounds}. There are specialized algorithms for this task \cite{GPRE95}.

\subsubsection{Finding the solution to \eqref{ryshproblatt}}
As described in Section \ref{boundssec}, the optimal $\vec{G}$ is one of the vertices in the polytope $\mathcal{R}_1 \cap \{\vec{G}\,:\,\mathrm{tr}(\vec{G}) \leq \mathrm{ub}(\vec{S})\}$. Hence, one method to find the optimum is to directly enumerate all the perfect forms inside the polytope. 
A finite codebook can be constructed off-line if a priori knowledge of the upper bound in \eqref{enub2} is available. A method to enumerate perfect forms is via Voronoi's algorithm \cite{Voronoi07}. It enumerates perfect forms and stops when all non-isometric forms have been found. As mentioned, for today's computers, it is only usable up to 8 dimensions due to the large number of edges in the Ryshkov polytope in higher dimensions. We need to slightly modify the classical Voronoi's algorithm, by changing its stopping condition. Since we are interested in forms that are isometric, our stopping condition is based on the upper bound in \eqref{enub2}. Further, we note that transforming $\vec{G}_{\vec{B}}$ to $\vec{\Pi}^{\transp}\vec{G}_{\vec{B}}\vec{\Pi}$, where $\vec{\Pi}$ is a generalized permutation matrix $\vec{\Pi}$ such that each non-zero element in $\vec{\Pi}$ is either $1$ or $-1$, does not change the value of the objective function in \eqref{ourproblatt2}, since the eigenvalues of $\vec{G}_{\vec{B}}$ are still the same. Moreover, this does not change the constraint region in \eqref{ourproblatt2} since it merely permutes the error vectors. For our algorithm, this observation implies that $\vec{G}$ matrices such that $\vec{G} = \vec{\Pi}^{\transp}\hat{\vec{G}}\vec{\Pi}$, for some already visited perfect form $\hat{\vec{G}}$, do not have to be traversed by the algorithm; thus, the search space can be reduced. The algorithm needs an initial perfect form as starting position, and a good starting point is the root lattice $A_N$ \cite{Achill}. The following notation is used in the algorithm. $\mathrm{Min}(\vec{G})$ denotes the set of minimum vectors of $\vec{G}$, i.e., the set $\{\vec{x}:\,\,\vec{x}^{\transp}\vec{G}\vec{x} = 1\}$ and $\vec{G}[\vec{x}] \stackrel{\triangle}{=} \vec{x}^{\transp}\vec{G}\vec{x}$. The algorithm is summarized by the pseudo-code in Table \ref{alg1}.
\begin{table}
\begin{minipage}[c]{\linewidth}
 \begin{center}
 \begin{tabular}{ p{84mm} }
 \hline\hline
\textbf{Algorithm 1}\\
\textsc{Input: A starting perfect form $\vec{G}_s$, e.g., the root lattice $A_N$.} \\ \hline 
\textsc{Output: The list of $\vec{G}$ matrices corresponding to the vertices in the polytope $\mathcal{R}_1 \cap \{\vec{G}\,:\,\vec{G} \leq \mathrm{ub}(\vec{S})\}$} \\
\hline \hline
Let $\vec{G} = \vec{G}_s$ and define the boolean variable $b_{\vec{G}} \stackrel{\triangle}{=} 1$. Save the pair $(\vec{G},b_{\vec{G}})$ in a set $\mathcal{G} = \{(\vec{G},b_{\vec{G}})\}$. 
\begin{enumerate}
\item Compute $\mathrm{Min}(\vec{G})$ and the extreme rays (edges) $\vec{T}_1,\ldots,\vec{T}_k$ of the polyhedal cone $$\{\vec{G}' \in \mathcal{S}^{N\times N}:\,\,\vec{G}'[\vec{x}] \geq 0 \,\, \forall \vec{x} \in \mathrm{Min}(\vec{G})\}.$$ 
\item Determine neighbouring perfect forms $\vec{G}_i$ as $\vec{G}_i = \vec{G} + \alpha \vec{T}_i$, $i = 1\ldots k$.
\item Let $\vec{G}_{j_1},\ldots \vec{G}_{j_m}$ be those neighbouring forms satisfying the upper bound in \eqref{enub2} and that cannot be expressed as $\vec{\Pi}^{\transp}\hat{\vec{G}}\vec{\Pi}$ for some $(\hat{\vec{G}},b_{\hat{\vec{G}}})$ in $\mathcal{G}$. Define $b_{\vec{G}_{j_l}} \stackrel{\triangle}{=} 0$, $l = 1 \ldots m$ and let $\mathcal{G} = \mathcal{G} \cup \{(\vec{G}_{j_1},b_{\vec{G}_{j_1}}),\ldots,(\vec{G}_{j_m},b_{\vec{G}_{j_m}})\}$.
\item Find a pair $(\vec{G}_j,b_{\vec{G}_j})$ in $\mathcal{G}$ such that $b_{\vec{G}_j} = 0$. If such a pair exists, change the value of $b_{\vec{G}_j}$ to $b_{\vec{G}_j} = 1$, let $\vec{G} = \vec{G}_j$ and go to step 1. Otherwise, stop and return $\mathcal{G}$.
\end{enumerate}
\\\hline 
\end{tabular}
 \end{center}
 \caption{An algorithm that traverses all perfect forms satisfying the bounds in \eqref{enub2}.}
\label{alg1} \end{minipage}
 \end{table}
The only difference between Algorithm 1 and the Voronoi algorithm presented in \cite{Voronoi07} is the stopping condition and search space reduction through generalized permutation matrices. Voronoi's algorithm stops as soon as all neighbouring perfect forms of a certain perfect form are isometric to some other perfect form already encountered. Algorithm \ref{alg1} stops as soon as all perfect forms (up to a generalized permutation matrix) satisfying the upper bound in \eqref{enub2} have been enumerated.

The Fincke-Pohst algorithm is used to compute $\mathrm{Min}(\vec{G})$ \cite{FP85}. The toughest part of the algorithm is to compute the extreme rays of a polytope specified by linear inequalities. This is the bottleneck of enumerating non-isometric perfect forms with Voronoi's algorithm and thereby solving the lattice packing problem in high dimensions. There exist methods that do this in $O(Nvd)$ time, where $N$ is the dimension, $d$ the number of non-redundant inequalities describing the polytope and $v$ is the number of vertices in the polytope \cite{AF99}.

For step 2, determining the neighbouring perfect forms can be done by the algorithm in \cite{Achill}[Algorithm 2, Chapter 3], which computes the $\alpha$ needed in step 2. In the other steps, we use boolean variables $b_{\vec{G}_j}$ to denote whether a vertex has been visited or not.

\subsection{Suboptimal Precoder Construction}
\label{subprecsec}
Finding the optimal solution is a computationally demanding task for today's computers, and suboptimal solutions are of interest. We base our suboptimal construction on $\vec{F}_{\mathrm{gmd}}$ from Section \ref{boundssec}.

It can be numerically verified that for $\vec{F}_{\mathrm{gmd}}$, the received lattice is not a Minkowski extreme lattice, and is thereby not optimal. Hence, the performance of $\vec{F}_{\mathrm{gmd}}$ can be improved by applying the result of Theorem \ref{thm2}. Let $\vec{B}_{\mathrm{gmd}} = \vec{S}\vec{F}_{\mathrm{gmd}}$ be the received lattice at the receiver, where $\vec{F}_{\mathrm{gmd}}$ is scaled so that $D^2_{\min}(\vec{B}_{\mathrm{gmd}}) = 1$. Now perform a Minkowski reduction on $\vec{B}_{\mathrm{gmd}}$ by using the Minkowski reduction algorithm \cite{Min05}, so that we can factor the basis matrix as $\vec{B}_{\mathrm{gmd}} = \vec{U}_{\mathrm{gmd}}\vec{L}_{\mathrm{gmd}}\vec{Z}_{\mathrm{gmd}}$ for some rotation $\vec{U}_{\mathrm{gmd}}$, Minkowski reduced lattice basis $\vec{L}_{\mathrm{gmd}}$ with $D^2_{\min}(\vec{L}_{\mathrm{gmd}}) = 1$, and unimodular $\vec{Z}_{\mathrm{gmd}}$. It then follows that $\vec{F}_{\mathrm{gmd}} = \vec{S}^{-1}\vec{U}_{\mathrm{gmd}}\vec{L}_{\mathrm{gmd}}\vec{Z}_{\mathrm{gmd}}$. Define $\vec{F}_{m,i} \stackrel{\triangle}{=} \vec{S}^{-1}\vec{U}_{i}\vec{L}_{m,i}\vec{Z}_{\mathrm{gmd}}$, $i = 1 \ldots K$, to be the $K$ different precoders where $\vec{L}_{m,i}$ is the $i$:th Minkowski extreme lattice with $D^2_{\min}(\vec{L}_{m,i}) = 1$, and $\vec{U}_i$ is the left orthogonal matrix of $\vec{L}_{m,i}\vec{Z}_{\mathrm{gmd}}$. Applying Theorem \ref{thm2}, we know that $\mathrm{tr}(\vec{F}_{m,j}^{\transp}\vec{F}_{m,j}) < \mathrm{tr}(\vec{F}_{\mathrm{gmd}}^{\transp}\vec{F}_{\mathrm{gmd}})$ for some $1 \leq j \leq K$. Thus, $\vec{F}_{m,j}$ is a precoder performing better than the geometric mean precoder.

Hence, by performing a Minkowski reduction and using the resulting unimodular matrix together with one of the Minkowski extreme lattices, it is possible to improve upon the geometric mean precoder and reach closer to the lower bound given in \eqref{enbounds}. However, performing a Minkowski reduction includes finding the shortest basis vector in the lattice, which is an NP-hard problem \cite{Micciancio}. Nevertheless, it turns out to be easily doable with a standard workstation at least for $N \lesssim 15$.
Another method that can be used for this purpose is the iterative algorithm presented in \cite{Svante}.

\subsection{Packing lattices are not always a solution to \eqref{ryshproblatt}}
\label{denscontrsec}
Note that a perfect form $\vec{G}$ that is a candidate for solving \eqref{ryshproblatt} must satisfy the upper bound \eqref{enub}. Since $\vec{e}^{\transp}\vec{G}\vec{e} \geq 1$, $\forall \vec{e} \in \mathbb{Z}^N/\{\vec{0}_N\}$, it holds that $g_{i,i} \geq 1$, $i = 1 \ldots N$, and thus $\mathrm{tr}(\vec{G}) \geq N$. Let $\lambda_i$ denote the $i$:th shortest vector in the lattice $\vec{L}$ with Gram matrix $\vec{G}$. By definition, the minimum distance is $\lambda_1 = 1$. Now assume a channel $\vec{S}$ such that the upper bound in \eqref{enub} is smaller than $(N-1)\lambda_1 + \lambda_2$. If $\vec{Z}^{\transp}\vec{G}\vec{Z}$ is another perfect form isometric to $\vec{G}$ that is also a candidate for solving \eqref{ryshproblatt}, then $\mathrm{tr}(\vec{Z}^{\transp}\vec{G}\vec{Z}) \leq (N-1)\lambda_1 + \lambda_2$. However, this inequality significantly limits the number of possible $\vec{Z}$ matrices and thus the number of perfect forms isometric to $\vec{G}$. Namely, since $\mathrm{tr}(\vec{Z}^{\transp}\vec{G}\vec{Z}) = \sum_{j=1}^N \vec{Z}_j^{\transp}\vec{G}\vec{Z}_j$ and each term $\lambda_1 \leq \vec{Z}_j^{\transp}\vec{G}\vec{Z}_j \leq \lambda_2$, it follows that each $\vec{Z}_j$ must correspond to a minimum vector of $\vec{G}$, i.e., $\vec{Z}_j$ belongs to the set $\mathrm{Min}(\vec{G})$. 

Let us apply this idea to 4-dimensional lattices. In 4 dimensions, there are only two non-isometric perfect forms, the $D_4$ and $A_4$ lattice. A Gram matrix for $D_4$\footnote{Gram matrices for non-isometric perfect forms can be found at \cite{Sloanehomepage}.} is 
\be
\label{GD4}
\vec{G}_{D_4} = \left(\begin{array}{cccc}
1 & 0 & 0.5 & 0 \\
0 & 1 & -0.5 & 0 \\
0.5 & -0.5 & 1 & -0.5 \\
0 & 0 & -0.5 & 1
\end{array}\right)
\ee
and for $A_4$, 
\be
\label{GA4}
\vec{G}_{A_4} = \left(\begin{array}{cccc}
1 & -0.5 & 0 & 0 \\
-0.5 & 1 & -0.5 & 0 \\
0 & -0.5 & 1 & -0.5 \\
0 & 0 & -0.5 & 1
\end{array}\right).
\ee
Hence, any perfect form in 4 dimensions can be expressed as either $\vec{Z}^{\transp}\vec{G}_{D_4}\vec{Z}$ or $\vec{Z}^{\transp}\vec{G}_{A_4}\vec{Z}$ for some unimodular $\vec{Z}$.
Further, it holds that $\lambda_1 = 1$ and $\lambda_2 = 2$ for both $\vec{G}_{A_4}$ and $\vec{G}_{D_4}$.
Now let $\vec{S}$ be 
\be
\label{Sexp}
\vec{S} = \left(\begin{array}{cccc}
1 & 0 & 0 & 0 \\
0 & 0.95 & 0 & 0 \\
0 & 0 & 0.94 & 0 \\
0 & 0 & 0 & 0.93
\end{array}\right).
\ee
The upper bound in \eqref{enub} is 4.83 for this $\vec{S}$. Hence, if a perfect form $\vec{Z}^{\transp}\vec{G}_{D_4}\vec{Z}$ isometric to $\vec{G}_{D_4}$ solves \eqref{ryshproblatt}, then the columns of $\vec{Z}$ must be taken from $\mathrm{Min}(\vec{G}_{D_4})$. Similarly, if a perfect form isometric to $\vec{G}_{A_4}$ solves \eqref{ryshproblatt}, then the columns of the corresponding unimodular $\vec{Z}$ are taken from $\mathrm{Min}(\vec{G}_{A_4})$. It is an easy task to find $\mathrm{Min}(\vec{G}_{D_4})$ and $\mathrm{Min}(\vec{G}_{A_4})$, by applying the Fincke-Pohst algorithm, and also to find all unimodular matrices whose columns consist of these minimum vectors. Going through each perfect form obtained from these unimodular matrices, and plugging in the optimal precoder \eqref{optFconstr} into \eqref{ryshproblatt}, the result is that the perfect form isometric to $\vec{G}_{A_4}$ gives the smallest value of the objective function in \eqref{ryshproblatt}. Repeating the same argument for the channel
\be
\label{Sexp2}
\vec{S} = \left(\begin{array}{cccc}
1 & 0 & 0 & 0 \\
0 & 0.99 & 0 & 0 \\
0 & 0 & 0.94 & 0 \\
0 & 0 & 0 & 0.93
\end{array}\right),
\ee
one concludes that the perfect form isometric to $\vec{G}_{D_4}$ solves \eqref{ryshproblatt}. Hence, this shows that both $A_4$ and $D_4$ occur as optimal lattice structures at the receiver; which one it is, depends on the channel $\vec{S}$. 
\section{Conclusions}
\label{conclsec}
This work studies precoding over non-singular linear channels with full CSI through a lattice-theoretic approach. The classical complex-valued linear channel is first transformed to a more general real-valued model which enables performance improvements over the classical complex-valued model. Then, the main problem studied in the work is to find lattices that maximize the minimum distance between the received lattice points, under an average energy constraint at the transmitter. The optimal lattice is analytically shown to be a perfect lattice, as defined by Ryshkov, for \emph{any} given non-singular channel. Bounds on the optimal performance are developed, tighter than previously reported, which enable construction of algorithms that produce a finite codebook of matrices, from which the optimal precoder can be derived. Furthermore, a suboptimal precoder construction is presented together with bounds on its performance, which is analytically shown to improve upon a previous presented precoding scheme in the literature, by utilizing the new results in this work. In addition to this, we demonstrate with an example that optimal packing lattices are not always optimal for maximizing minimum distance, which is a counter-intuitive result at first sight. An immediate practical application of the derived results is precoding over large alphabets.
\label{concl}
 
\end{document}